\documentclass[11pt, a4paper]{article}

\usepackage{amsmath}
\usepackage{amssymb}
\usepackage{amsfonts}
\usepackage{amsthm}
\usepackage{graphicx}
\usepackage{fullpage}
\usepackage{hyperref}
\usepackage{cite}
\usepackage{enumitem}
\usepackage{thm-restate}
\usepackage{tabularx}

\newtheorem{theorem}{Theorem}[section]
\newtheorem{lemma}[theorem]{Lemma}
\newtheorem{proposition}[theorem]{Proposition}

\newtheorem{claim}[theorem]{Claim}

\newtheorem{example}[theorem]{Example}

\usepackage{xcolor}

\usepackage{tabularx}
\usepackage {environ}

\makeatletter
\newcommand{\problemtitle}[1]{\gdef\@problemtitle{#1}}
\newcommand{\probleminput}[1]{\gdef\@probleminput{#1}}
\newcommand{\problemoutput}[1]{\gdef\@problemoutput{#1}}
\NewEnviron{problem}{
  \problemtitle{}\probleminput{}\problemoutput{}
  \BODY
  \par\addvspace{.5\baselineskip}
  \noindent{
    \framebox[\textwidth][c]{
        \begin{tabularx}{\textwidth}{@{\hspace{\parindent}} l X c}
            \multicolumn{2}{@{\hspace{\parindent}}l}{\textsc{\@problemtitle}} \\
            \textbf{Input:} & \@probleminput \\
            \textbf{Question:} & \@problemoutput
        \end{tabularx}
        }
      }
  \par\addvspace{.5\baselineskip}
}
\makeatother

\title{Reconfiguration of the Union of Arborescences\thanks{A conference version of this paper appears in ISAAC 2023~\cite{isaac2023}, 
which is selected as one of the best papers in ISAAC 2023 by the program committee.}}

\author{Yusuke Kobayashi\thanks{Research Institute for Mathematical Sciences, Kyoto University.
E-mail: yusuke@kurims.kyoto-u.ac.jp}
\and
Ryoga Mahara\thanks{Department of Mathematical Informatics, University of Tokyo.
E-mail: mahara@mist.i.u-tokyo.ac.jp}
\and
Tam\'{a}s Schwarcz\thanks{MTA-ELTE Momentum Matroid Optimization Research Group, Department of Operations Research, ELTE E\"{o}tv\"{o}s Lor\'{a}nd University. E-mail: tamas.schwarcz@ttk.elte.hu}
}

\date{}

\begin{document}

\maketitle

\begin{abstract}
An arborescence in a digraph is an acyclic arc subset in which every vertex except a root has exactly one incoming arc. 
In this paper, we show the reconfigurability of the union of $k$ arborescences for fixed $k$ 
in the following sense: 
for any pair of arc subsets that can be partitioned into $k$ arborescences, 
one can be transformed into the other by exchanging arcs one by one
so that every intermediate arc subset can also be partitioned into $k$ arborescences. 
This generalizes the result by Ito et al.~(2023), who showed the case with $k=1$. 
Since the union of $k$ arborescences can be represented as a common matroid basis of two matroids, 
our result gives a new non-trivial example of matroid pairs for which 
two common bases are always reconfigurable to each other. 
\end{abstract}

\section{Introduction}

\subsection{Reconfigurability of Common Bases of Matroids}

Exchanging a pair of elements, i.e., adding one element to a set and removing another element from it, is a fundamental operation in matroid theory. 
The basis exchange axiom for matroids implies that,  
for any pair of bases of a matroid,  
one can be transformed into the other by repeatedly exchanging pairs of elements
so that all the intermediate sets are also bases. 
That is, the basis family of a matroid is connected with respect to element exchanges. 
This is an important property of matroid basis families that is used in various contexts, 
e.g., it is a key to show the validity of a local search algorithm
for finding a maximum weight basis. 

In contrast to matroid basis families, a family of common bases of two matroids does not necessarily enjoy this property. 
More precisely, for two matroids $M_1$ and $M_2$ over a common ground set, 
the following condition, which we call {\it Reconfigurability of Common Bases (RCB)}, does not necessarily hold. 
\begin{enumerate}[label=(RCB),align=parleft,labelindent=0pt,itemindent=0pt,labelsep=5pt,leftmargin=*]
\item \label{it:rcb} For any pair of common bases $B$ and $B'$ of two matroids $M_1$ and $M_2$,
there exists a sequence of common bases $B_0, B_1, \dots , B_\ell$ such that $B_0 = B$, $B_\ell = B'$,  
$B_i$ is a common basis, and $|B_{i-1} \setminus B_i| = |B_{i} \setminus B_{i-1}| = 1$ for each $i \in \{1,2, \ldots , \ell\}$.
\end{enumerate}

As an example, suppose that $G$ is a cycle of length four, which has exactly two perfect matchings. 
Since $G$ is a bipartite graph, the family of all the perfect matchings in $G$ 
can be represented as a family of common bases of two matroids, 
and we see that  \ref{it:rcb} does not hold in this setting. 

On the other hand, there are some special cases satisfying \ref{it:rcb}. 
When $M_1$ is a graphic matroid and $M_2$ is the dual matroid of $M_1$, 
it is known that \ref{it:rcb} holds~\cite{farber1985edge}. 
A conjecture by White~\cite{white1980unique} that \ref{it:rcb} holds when $M_1$ is an arbitrary matroid and $M_2$ is its dual has been open for more than 40 years. The conjecture was verified for strongly base orderable matroids~\cite{lason2014toric} and for sparse paving matroids~\cite{bonin2013basis}. Recently, the proof of the graphic case was extended to regular matroids~\cite{berczi2023reconfiguration}, moreover, the conjecture was settled for split matroids~\cite{berczi2023exchange}, a large class containing paving matroids as well. 

Another special case with property \ref{it:rcb} is the family of  all arborescences in a digraph. 
For a digraph $D=(V, A)$ with a specified vertex $r \in V$ called a \emph{root}, 
an \emph{$r$-arborescence} is an acyclic arc subset of $A$ in which every vertex in $V \setminus \{r\}$ has exactly one incoming arc. 
When the root vertex is not specified, it is simply called an \emph{arborescence}.  
We see that an $r$-arborescence (or an arborescence) is represented as a common basis of the graphic matroid corresponding to the acyclic constraint and the partition matroid corresponding to the indegree constraint. 
It is shown by Ito et al.~\cite{ito2023reconfiguring} that the family of all arborescences (or $r$-arborescences) in a digraph satisfies \ref{it:rcb}.

\subsection{Our Results}

In this paper, we mainly study the union of $k$ arc-disjoint $r$-arborescences, where $k$ is a fixed positive integer. 
For a digraph $D=(V, A)$ and a root $r \in V$, 
let $\mathcal{F}_{k, r} \subseteq 2^A$ denote the set of all arc subsets that 
can be partitioned into $k$ arc-disjoint $r$-arborescences.
It is known that $\mathcal{F}_{k, r}$ is represented as the common bases of two matroids $M_1$ and $M_2$, 
where $M_1$ is the union of $k$ graphic matroids and $M_2$ is the direct sum of uniform matroids corresponding to the indegree constraint; see~\cite[Corollary 53.1c]{lexbook}.  
The main contribution of this paper is to show that such a pair of $M_1$ and $M_2$ satisfies property \ref{it:rcb}. 
Formally, our main result is stated as follows. 

\begin{theorem}
\label{thm:main}
Let $D=(V, A)$ be a digraph with a root $r \in V$ and let $k$ be a positive integer. 
Let $\mathcal{F}_{k, r} \subseteq 2^A$ denote the family of all arc subsets that 
can be partitioned into $k$ arc-disjoint $r$-arborescences.
For any $S, T \in \mathcal{F}_{k, r}$,  
there exists a sequence 
$T_0, T_1, \dots , T_\ell$ such that $T_0 = S$, $T_\ell = T$,  
$T_i \in \mathcal{F}_{k, r}$ for $i \in \{0,1, \ldots , \ell\}$, and $|T_{i-1} \setminus T_i| = |T_{i} \setminus T_{i-1}| = 1$ for $i \in \{1,2, \ldots , \ell\}$. 
Furthermore, such a sequence can be found in polynomial time. 
\end{theorem}

We also prove the analogue of Theorem~\ref{thm:main} in the case when a feasible arc set is the union of $k$ arc-disjoint arborescences that may have distinct roots.
Formally, our result is stated as follows and will be discussed in Section~\ref{sec:distinctroot}. 
\begin{theorem}
\label{thm:main2}
Let $D=(V, A)$ be a digraph and let $k$ be a positive integer. 
Let $\mathcal{F}_{k} \subseteq 2^A$ denote the family of all arc subsets that 
can be partitioned into $k$ arc-disjoint arborescences.
For any $S, T \in \mathcal{F}_{k}$,  
there exists a sequence 
$T_0, T_1, \dots , T_\ell$ such that $T_0 = S$, $T_\ell = T$,  
$T_i \in \mathcal{F}_{k}$ for $i \in \{0,1, \ldots , \ell\}$, and $|T_{i-1} \setminus T_i| = |T_{i} \setminus T_{i-1}| = 1$ for $i \in \{1,2, \ldots , \ell\}$. 
Furthermore, such a sequence can be found in polynomial time. 
\end{theorem}

When $k=1$, Theorem~\ref{thm:main} amounts to the reconfigurability of $r$-arborescences, which is an easy result; see e.g.~\cite{barahona1987exact,ito2023reconfiguring}. 
In this special case, we can update an $r$-arborescence $S$ so that $|S\setminus T|$ decreases monotonically, 
which immediately leads to the existence of a reconfiguration sequence of length $|S\setminus T|$. 
Here, the \emph{length} of a reconfiguration sequence is defined as the number of exchange operations in it. 
Meanwhile, such an update of $S$ is not always possible when $k \ge 2$, 
that is, there exists an example in which more than $|S  \setminus T|$ steps are required to 
transform $S$ into $T$;  see Example~\ref{ex:detour}. 
This suggests that the case with $k\ge 2$ is much more complicated than the case with $k=1$.

\begin{example}
\label{ex:detour}
Let $k=2$, let $D=(V, A)$ be a digraph with a root $r \in V$, and let $S, T \in \mathcal{F}_{k, r}$ be as in Figure~\ref{fig:example1}. 
Then, the shortest reconfiguration sequence between $S$ and $T$ has length $3$, while $|S \setminus T| = |T \setminus S| =2$. 
\end{example}
\begin{figure}
    \centering
    \includegraphics[width=130mm]{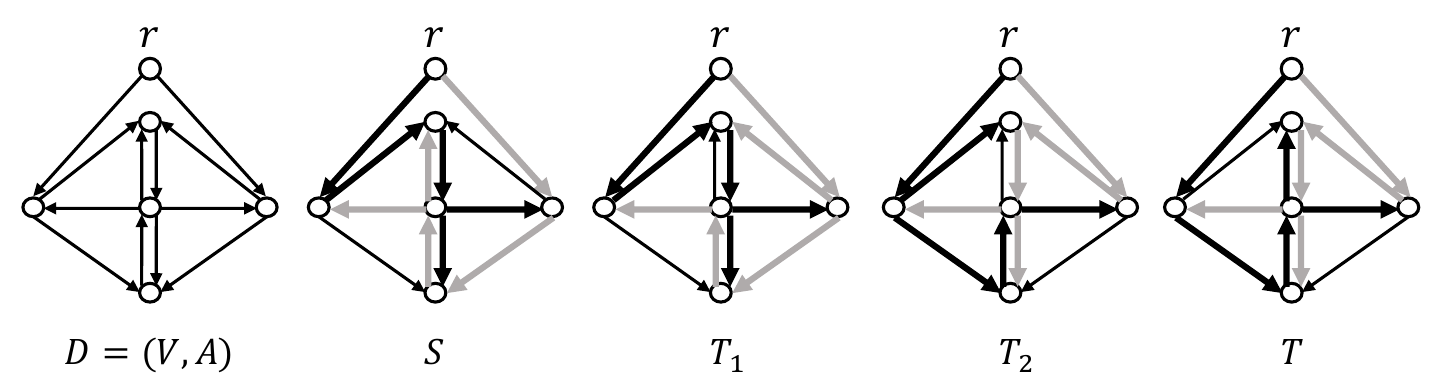}
    \caption{The leftmost figure represents a digraph $D=(V,A)$ with a root $r\in V$.
    In each figure, the union of thick black and gray arc subsets is in $\mathcal{F}_{k, r}$, where $k=2$.
    The sequence $T_0=S, T_1, T_2, T_3=T$ is a shortest reconfiguration sequence from $S$ to $T$.
    }
    \label{fig:example1}
\end{figure}

We here give another remark on the relationship between the case of $k\ge 2$ and the case of $k=1$. 
We have already mentioned that 
an $r$-arborescence is represented as a common basis of the graphic matroid $M_1$ 
and the partition matroid $M_2$ corresponding to the indegree constraint. 
Then, $\mathcal{F}_{k, r}$ is represented as a family of common bases of $k M_1$ and $k M_2$, 
where $k M_i$ is the matroid whose bases are the unions of $k$ disjoint bases of $M_i$. 
Since the matroids representing $\mathcal{F}_{k, r}$ have these special forms, to prove Theorem~\ref{thm:main}, 
it will be natural to expect that the case of $k \ge 2$ is reduced to the case of $k=1$. 
However, the following theorem suggests that such an approach does not work naively; see Section~\ref{sec:sumcounterexample} for the proof.

\begin{restatable}{theorem}{sumcounterexample}
\label{thm:sumcounterexample}
There exist matroids $M_1= (E, \mathcal{B}_1)$ and $M_2= (E, \mathcal{B}_2)$ such that 
the pair of matroids $(M_1, M_2)$ satisfies \ref{it:rcb}, while $(2 M_1, 2 M_2)$ does not satisfy it. 
\end{restatable}

In this paper, we also consider the following algorithmic problem
in addition to the reconfigurability of the union of arborescences. 

\begin{problem}
\problemtitle{RCB Testing}
\probleminput{Two matroids $M_1$ and $M_2$.}
\problemoutput{Determine whether $M_1$ and $M_2$ satisfy \ref{it:rcb}.}
\end{problem}

By constructing a hard instance, we show the hardness result for this problem; see Section~\ref{sec:MIRHardness} for the proof.

\begin{restatable}{theorem}{hardRCB}
    \label{thm:hardRCB}
    {\sc RCB Testing} requires an exponential number of independence queries
    if the matroids are given by independence oracles. 
\end{restatable}
It is worth noting that the proof of this theorem implies the hardness of 
another problem called {\sc Matroid Intersection Reconfiguration} (see Section~\ref{sec:related} for the problem description), 
resolving an open question mentioned in \cite{Bousquet2023TCS}.

\subsection{Related Work}
\label{sec:related}

The study of packing arborescences was initiated by Edmonds~\cite{edmonds1973edge}, 
who showed that a digraph $D=(V, A)$ contains $k$ arc-disjoint $r$-arborescences if and only if 
every nonempty subset of $V \setminus \{r\}$ has at least $k$ entering arcs; see Theorem~\ref{thm:edmonds}. 
Lov\'asz~\cite{Lovasz1976} gave a simpler proof for this theorem. 
There are several directions of extension of Edmonds' theorem. 
The first one by Frank~\cite{Frank1978mixed} is to extend directed graphs to mixed graphs. 
Second, Frank, Kir\'aly, and Kir\'aly~\cite{FrankKK2003} extended directed graphs to directed hypergraphs. 
Third, Kamiyama, Katoh, and Takizawa~\cite{KamiyamaKT09} and Fujishige~\cite{Fujishige10}
extended the problem to the packing of rooted-trees that cover only reachable vertex set. 
Lastly, de Gevigney, Nguyen, and Szigeti~\cite{GNZ13} considered the packing of rooted-trees with matroid constraints.
By combining these extensions, we can consider further generalizations, 
which have been actively studied in~\cite{BF2008,FORTIER2018,GAO2021313,Csaba2016,MATSUOKA20191}.

{Combinatorial reconfiguration} is an emerging field in discrete mathematics and theoretical computer science. 
One of the central problems in combinatorial reconfiguration is the following algorithmic question; 
for two given discrete structures, determine whether one can be transformed into the other by a sequence of local changes. 
See surveys of Nishimura~\cite{DBLP:journals/algorithms/Nishimura18} and van den Heuvel \cite{DBLP:books/cu/p/Heuvel13}, 
and see also solvers for combinatorial reconfiguration~\cite{ito2023zddbased}.

In this framework, if we focus on common bases of two matroids, then we can consider the following reconfiguration problem. 
\begin{problem}
\problemtitle{Matroid Intersection Reconfiguration}
\probleminput{Two matroids $M_1$ and $M_2$ and their common bases $B$ and $B'$.}
\problemoutput{Determine whether $B$ can be transformed into $B'$ by exchanging a pair of elements repeatedly 
                so that all the intermediate sets are common bases of $M_1$ and $M_2$. }
\end{problem}
Although this problem seems to be a fundamental problem, its polynomial solvability was previously unknown (see \cite[Question 4]{Bousquet2023TCS}). 
If the pair of input matroids $M_1$ and $M_2$ satisfies~\ref{it:rcb}, then {\sc Matroid Intersection Reconfiguration} is trivial, i.e., the transformation is always possible. 
If each common basis corresponds to a maximum matching in a bipartite graph, then 
the pair of $M_1$ and $M_2$ does not necessarily satisfy~\ref{it:rcb}, but {\sc Matroid Intersection Reconfiguration} can be solved in polynomial time~\cite{DBLP:journals/tcs/ItoDHPSUU11}. 
Actually, it is shown in~\cite{DBLP:journals/tcs/ItoDHPSUU11} that the reconfiguration problem of maximum matchings in non-bipartite graphs 
is also solvable in polynomial time. 

Since the spanning trees in a graph form a matroid basis family, the reconfiguration problem on spanning trees is trivial, i.e., the transformation is always possible.  
However, the problem becomes non-trivial if we add some constraints. 
Reconfiguration problems on spanning trees with additional constraints were studied in~\cite{Bousquet2020,Bousquet2023Algo}.

\subsection{Overview}

We now describe an outline of the proof of Theorem~\ref{thm:main}.  
For feasible arc subsets $S$ and $T$, in order to show that $S$ can be transformed into $T$, 
it suffices to show that $S$ can be transformed into another feasible arc subset $S'$ such that $|S' \setminus T| = |S \setminus T|-1$. 
When $k=1$, such $S'$ can be easily obtained from $S$ by exchanging only one pair of arcs; see~\cite{barahona1987exact,ito2023reconfiguring}. 
However, this is not indeed the case when $k \ge 2$ as shown in Example~\ref{ex:detour}, that is, 
several steps may be required to obtain $S'$. This is the main technical difficulty of the problem. 

To overcome this difficulty, we introduce and use \emph{minimal tight sets} and an \emph{auxiliary digraph}. 
Let $T \setminus S = \{f_1, \dots , f_p\}$. 
For each arc $f_i \in T \setminus S$, we consider the inclusionwise minimal vertex set $X_i$ 
subject to $X_i$ contains $f_i$ and exactly $k$ arcs in $S$ enter $X_i$ (i.e., $X_i$ is \emph{tight}).
Then, $X_i$ gives a characterization of arcs $e \in S$ that can be replaced with $f_i$; see Lemma~\ref{lem:01}.
By using $X_1, \dots , X_p$, we construct an auxiliary digraph $H$, whose definition is given in Section~\ref{sec:auxiliary}.  
Roughly speaking, $H$ is similar to the exchangeability digraph when $\mathcal{F}_{k,r}$ is represented as the common bases of two matroids. 
Then, we can show that $H$ has a dicycle; see Lemma~\ref{lem:04}. 
If $H$ has a self-loop, then we can obtain a desired arc subset $S'$ by exchanging only one pair of arcs. 
Otherwise, we update the arc set $S$ so that the length of the shortest dicycle in $H$ becomes shorter, which is discussed in Section~\ref{sec:shortestdicycle}. 
By repeating this procedure, we obtain a desired arc subset $S'$ in finite steps.

The rest of the paper is organized as follows. 
In Section~\ref{sec:preliminaries}, we introduce some notation and show basic results on arborescence packing. 
In Section~\ref{sec:proofmain}, we give a proof of Theorem~\ref{thm:main}, which is the main part of the paper. 
The upper bound on the length of a shortest reconfiguration sequence is discussed in Section~\ref{sec:upperbound}. 
In Section~\ref{sec:distinctroot}, we show the analogue of Theorem~\ref{thm:main} in the case when arborescences may have distinct roots and give a proof of Theorem~\ref{thm:main2}. 
In Section~\ref{sec:sumcounterexample}, we show that \ref{it:rcb} is not closed under sums (Theorem~\ref{thm:sumcounterexample}) by giving a concrete example. 
In Section~\ref{sec:exactpoly}, as an application of Theorem~\ref{thm:main2}, we present a polynomial-time algorithm 
for the so-called exact arborescence packing problem. 
Hardness of {\sc RCB Testing} and {\sc Matroid Intersection Reconfiguration} is discussed in Section~\ref{sec:MIRHardness}. 
Finally, in Section~\ref{sec:conclusion}, we conclude this paper by giving some remarks.

\section{Preliminaries}
\label{sec:preliminaries}

Let $D=(V, A)$ be a digraph that may have parallel arcs.   
For an arc $e \in A$, the head and tail of $e$ are denoted by ${\rm head} (e)$ and  ${\rm tail}(e)$, respectively. 
For $e \in A$ and $X \subseteq V$, we say that $X$ \emph{contains} $e$ if $X$ contains both ${\rm head} (e)$ and  ${\rm tail}(e)$.  
For $F \subseteq A$ and $X, Y \subseteq V$, let $\Delta_F(X, Y)$ denote the set of arcs in $F$ from $X$ to $Y$, 
i.e., $\Delta_F(X, Y) = \{ e \in F \mid {\rm tail} (e) \in X,\ {\rm head} (e) \in Y\}$. 
Let $\Delta^-_F(X)$ denote $\Delta_F(V \setminus X, X)$. 
Let $\delta_F(X, Y) = |\Delta_F(X, Y)|$ and $\delta^-_F(X) = |\Delta^-_F(X)|$. 
For $v \in V$, $\delta^-_F(\{v\})$ is simply denoted by $\delta^-_F(v)$.  

For a digraph $D=(V, A)$ with a specified vertex $r \in V$ called a \emph{root}, 
an \emph{$r$-arborescence} is an acyclic arc subset $T \subseteq A$ such that 
$\delta^-_T(v) = 1$ for $v \in V \setminus \{r\}$ and $\delta^-_T(r) = 0$. 
Note that an arborescence is often defined as a subgraph of $D$ in the literature, but it is regarded as an arc subset in this paper. 
For a positive integer $k$, 
let $\mathcal{F}_{k, r} \subseteq 2^A$ denote the family of all arc subsets that 
can be partitioned into $k$ arc-disjoint $r$-arborescences.
If $r$ and $k$ are clear, an arc subset in $\mathcal{F}_{k, r}$ is simply called \emph{feasible}. 
Note that for any feasible arc subsets $S$ and $T$, $|S|=|T|$ holds.
Edmonds~\cite{edmonds1973edge} gave the following characterization of feasible arc subsets. 
\begin{theorem}[\mbox{Edmonds~\cite{edmonds1973edge}}]
\label{thm:edmonds}
 For a digraph $D=(V, A)$ with $r \in V$, an arc subset $T \subseteq A$ is in $\mathcal{F}_{k, r}$ if and only if 
$\delta^-_T(v) = k$ for any $v \in V \setminus \{r\}$, $\delta^-_T(r) =0$, and  $\delta^-_T(X) \ge k$ for any $X \subseteq V \setminus \{r\}$ with $X \neq \emptyset$. 
\end{theorem}
For a feasible arc subset $T \subseteq A$, 
we say that a vertex set $X \subseteq V \setminus \{r\}$ is \emph{tight with respect to $T$} if $\delta^-_{T}(X) = k$. 
It is well-known that the tight sets are closed under intersection and union as follows. 

\begin{lemma}
\label{lem:tight}
Let $T\subseteq A$ be a feasible arc set and let $X, Y \subseteq V \setminus \{r\}$ be tight sets with respect to $T$ with $X \cap Y \neq \emptyset$. 
Then, $X \cap Y$ and $X \cup Y$ are tight sets with respect to $T$. 
Furthermore, $T$ has no arc connecting $X \setminus Y$ and $Y \setminus X$. 
\end{lemma}

\begin{proof}
By a simple counting argument, we obtain 
\begin{align*}
k + k 
&= \delta^-_T (X) + \delta^-_T (Y) \\
&= \delta^-_T (X \cap Y) + \delta^-_T (X \cup Y)  +  \delta_T(X \setminus Y, Y \setminus X) + \delta_T(Y \setminus X, X \setminus Y) \\ 
&\ge k + k + 0 + 0, 
\end{align*}
which shows that 
$\delta^-_T (X \cap Y) = \delta^-_T (X \cup Y) = k$ and 
$\delta_T(X \setminus Y, Y \setminus X) = \delta_T(Y \setminus X, X \setminus Y) = 0$. 
This means that $X \cap Y$ and $X \cup Y$ are tight and $T$ has no arc connecting $X \setminus Y$ and $Y \setminus X$. 
\end{proof}

For a positive integer $p$, let $[p] = \{1, 2, \dots , p\}$. 
For feasible arc subsets $S, T \in \mathcal{F}_{k, r}$,  
we say that $T_0, T_1, \dots , T_{\ell}$ is a \emph{reconfiguration sequence between $S$ and $T$} 
if $T_0 = S$, $T_{\ell} = T$, $T_i \in \mathcal{F}_{k, r}$ for any $i \in [\ell] \cup \{0\}$, and $|T_{i-1} \setminus T_i| = |T_{i} \setminus T_{i-1}| = 1$ for any $i \in [\ell]$. 
We call $\ell$ the \emph{length} of the sequence. 
With this terminology, Theorem~\ref{thm:main} is rephrased as follows: 
for any $S, T \in \mathcal{F}_{k, r}$, there always exists a reconfiguration sequence between $S$ and $T$.

We denote a matroid on ground set $E$ with family of bases $\mathcal{B}$ by $M=(E,\mathcal{B})$. 
See~\cite{oxley2011matroid} for the definition and basic properties of matroids. A rank-$r$ matroid is called \emph{paving} if it has no circuits of size less than $r$. A paving matroid is \emph{sparse paving} if its dual is also paving. Equivalently, a sparse paving matroid is a rank-$r$ matroid in which every set of size $r$ is either a basis or a circuit-hyperplane. The following technical statement is well-known, see e.g.~\cite{frank2011connections}.
\begin{lemma} \label{lem:paving}
    Let $r$ be a nonnegative integer and $E$ a set of size at least $r$. Let $\mathcal{H}$ be a (possibly empty) family of size-$r$ subsets of $E$ such that $|H\cap H'| \le r-2$ holds for each $H, H' \in \mathcal{H}$, $H \ne H'$. Then $\mathcal{B}_\mathcal{H} = \{B \subseteq E \mid |B|=r, B \not \in \mathcal{H}\}$ forms the family of bases of a sparse paving matroid. Moreover, every sparse paving matroid can be obtained in this form.  
\end{lemma}

\section{Proof of Theorem~\ref{thm:main}}
\label{sec:proofmain}

In this section, we prove Theorem~\ref{thm:main}.
For a digraph $D=(V,A)$ with a root $r$, let $S,T\in \mathcal{F}_{k, r}$ be feasible arc subsets.
Let $S \setminus T = \{e_1, \dots , e_p\}$ and  $T \setminus S = \{f_1, \dots , f_p\}$, where  $p = |S \setminus T|$.  
By changing the indices if necessary, we may assume that ${\rm head}(e_i) = {\rm head}(f_i)$ for any $i \in [p]$. 
To prove Theorem~\ref{thm:main}, it suffices to show that we can transform $S$ to a new feasible arc subset $S'$ such that $|S' \setminus T| = p-1$, 
which is formally stated as follows. 

\begin{proposition}
\label{prop:reducedifference}
Let $S \subseteq A$ and $T \subseteq A$ be feasible arc subsets with $|S \setminus T| = p$. 
Then, there is a new feasible arc subset $S' \subseteq A$ such that  
$|S' \setminus T| = p-1$ and 
there is a reconfiguration sequence between $S$ and $S'$. 
\end{proposition}

In what follows in this section, we give a proof of Proposition~\ref{prop:reducedifference} and prove Theorem~\ref{thm:main}. 
In Section~\ref{sec:keylemma}, we introduce a \emph{minimal tight set} $X_i$ for each $i \in [p]$ and 
show some properties of $X_i$. 
In Section~\ref{sec:auxiliary}, we construct an \emph{auxiliary digraph} using $X_i$ and show its properties. 
In Section~\ref{sec:shortestdicycle}, 
we show that $S$ can be modified so that the shortest dicycle length in the auxiliary digraph becomes shorter
until the desired $S'$ is found. 
Finally, we prove Proposition~\ref{prop:reducedifference} and Theorem~\ref{thm:main} in Section~\ref{sec:together}.

\subsection{Minimal Tight Sets}
\label{sec:keylemma}

For $i \in [p]$, define $X_i \subseteq V$ as the inclusionwise minimal tight vertex set with respect to $S$ that contains $f_i$. 
For notational convenience, define $X_i = V$ if no such tight set exists. 
Note that such $X_i$ is uniquely defined 
since tight sets are closed under intersection; see Lemma~\ref{lem:tight}. 
\begin{example}
\label{ex:Xi}
    Let $k$, $D=(V, A)$, $S$, and $T$ be as in Example~\ref{ex:detour}. 
    Then, $X_1$ and $X_2$ are as shown in Figure~\ref{fig:figure7}. 
\end{example}
We show some properties of $X_i$.

\begin{lemma}\label{lem:01}
Let $i \in [p]$. If $e \in S$ satisfies ${\rm head}(e) = {\rm head}(f_i)$ and $e$ is contained in $X_i$, 
then $S' = S - e + f_i$ is feasible. 
\end{lemma}

\begin{proof}
Since 
$\delta^-_{S'}(v) = k$ for any $v \in V \setminus \{r\}$ and  $\delta^-_{S'}(r) =0$, by Theorem~\ref{thm:edmonds}, 
it suffices to show that 
$\delta^-_{S'} (X) \ge k$ holds for any $X \subseteq V \setminus \{r\}$ with $X \neq \emptyset$. 

Assume to the contrary that there exists a nonempty subset $X$ of  $V \setminus \{r\}$ with $\delta^-_{S'} (X) < k$. 
Since $\delta^-_{S} (X) \ge k$, we obtain  $\delta^-_{S} (X) = k$, $e \in \Delta^-_{S} (X)$, and $f_i \not\in \Delta^-_{A} (X)$. 
Hence, ${\rm head}(e) = {\rm head}(f_i) \in X$, ${\rm tail}(e) \not\in X$, and ${\rm tail}(f_i) \in X$ (Figure~\ref{fig:figure2}). 
This shows that $X$ is a tight set with respect to $S$ containing $f_i$. 
By Lemma~\ref{lem:tight}, $Y := X \cap X_i$ is also a tight set with respect to $S$ containing $f_i$, 
which contradicts the minimality of $X_i$ as ${\rm tail}(e) \in X_i \setminus Y$. 
\end{proof}

\begin{figure}[tbp]
 \begin{tabular}{cc}
 \begin{minipage}[t]{0.5\hsize}
  \begin{center}
   \includegraphics[width=40mm]{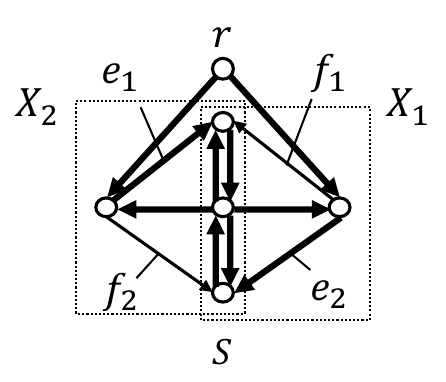}
   \caption{Minimal tight sets $X_i$.}
  \label{fig:figure7}
   \end{center}
 \end{minipage}
 \hfill
 \begin{minipage}[t]{0.5\hsize}
   \begin{center}
   \includegraphics[width=35mm]{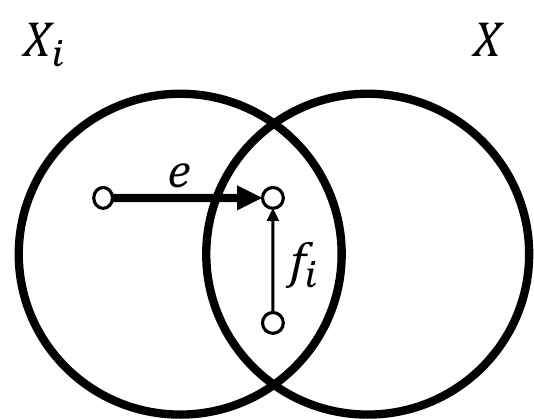}
   \caption{Proof of Lemma~\ref{lem:01}.}
     \label{fig:figure2}
    \end{center}
 \end{minipage}\\
 \end{tabular}
\end{figure}

Suppose that $f'_1 \in S \cap T$ is an arc such that  
${\rm head}(f'_1) = {\rm head}(f_1)$ and $f'_1$ is contained in $X_1$. 
By Lemma~\ref{lem:01}, $S' := S - f'_1 + f_1$ is feasible. 
For each $i \in [p]$, define $X'_i \subseteq V$ as the counterpart of $X_i$ associated with $S'$, i.e.,  
define $X'_i$ as the inclusionwise minimal tight vertex set with respect to $S'$ that contains $f_i$ ($f'_i$ if $i=1$), 
and $X'_i = V$ if no such set exists.

\begin{lemma}\label{lem:02}
Let $f'_1$ and $X'_1$ be as above. Then, it holds that $X'_1 = X_1$. 
\end{lemma}

\begin{proof}
We first show $X'_1 \subseteq X_1$. If $X_1 = V$, then $X'_1 \subseteq V = X_1$ is obvious. 
Otherwise, since  $X_1$ is a tight set with respect to $S$ that contains both $f_1$ and $f'_1$, 
we obtain $\delta^-_{S'}(X_1) = \delta^-_{S}(X_1) = k$, that is, $X_1$ is a tight set with respect to $S'$. 
This shows that $X'_1 \subseteq X_1$ by the minimality of $X'_1$.  

We next show $X'_1 \supseteq X_1$. If $X'_1 = V$, then $X'_1 = V \supseteq X_1$ is obvious. 
Otherwise, since  $X'_1$ is a tight set with respect to $S'$ that contains $f'_1$, 
we obtain $k = \delta^-_{S'}(X'_1)  \ge \delta^-_{S}(X'_1) \ge k$, which shows that $X'_1$ is a tight set with respect to $S$. 
This shows that $X'_1 \supseteq X_1$ by the minimality of $X_1$.  
This completes the proof. 
\end{proof}

Note that Lemma~\ref{lem:02} shows that the roles of $S$ and $S'$ are symmetric by replacing $f_1$ with $f'_1$. 
The following lemma shows a relationship between $X_i$ and $X'_i$, which plays a key role in our argument.

\begin{lemma}\label{lem:03}
Let $f'_1$ be an arc and $X'_i$ be the vertex set for $i \in [p]$ as above. 
Let $e \in S$ be an arc contained in $X_1$. 
For $i \in [p]$, we have one of the following: 
\begin{enumerate}
\item
$X_i = X'_i$,  
\item
$X_i$ contains $e$, or
\item
$X'_i$ contains $e$. 
\end{enumerate}
\end{lemma}

\begin{proof}
By Lemma~\ref{lem:02}, it suffices to consider the case when $i \in [p] \setminus \{1\}$. 
If $X_i = V$ or $X'_i = V$, then the second or third condition holds, and so 
we may assume that $X_i, X'_i \subseteq V \setminus \{r\}$. 

Assume that $X_i \neq X'_i$. 
Since the roles of $S$ and $S'$ are symmetric as we have seen in Lemma~\ref{lem:02}, 
without loss of generality, we may assume that $X_i \setminus X'_i \neq \emptyset$. 
Since $X_i$ is the inclusionwise minimal tight set containing $f_i$ with respect to $S$, 
$X_i \cap X'_i$ is not a tight set with respect to $S$, i.e., 
\begin{equation}
\delta^-_{S} (X_i \cap X'_i) \ge k+1,    \label{eq:01}
\end{equation}
where we note that $X_i \cap X'_i \neq \emptyset$ as both $X_i$ and $X'_i$ contain $f_i$; see Figure~\ref{fig:figure3} (left).  
We also see that 
\begin{equation}
\delta^-_{S} (X'_i) \le \delta^-_{S'} (X'_i) + 1 = k+1,  \label{eq:02}
\end{equation}
because $\delta^-_{S'} (X'_i) = k$ and $S' = S - f'_1 + f_1$. 

By (\ref{eq:01}) and (\ref{eq:02}), we obtain 
\begin{align*}
k + (k+1) 
&\ge \delta^-_{S} (X_i) + \delta^-_{S} (X'_i) \\ 
&= \delta^-_{S} (X_i \cap X'_i) + \delta^-_{S} (X_i \cup X'_i) + \delta_S (X_i \setminus X'_i, X'_i \setminus X_i) + \delta_S (X'_i \setminus X_i, X_i \setminus X'_i) \\
&\ge (k+1) + k + 0 + 0.  
\end{align*}
This shows that all the inequalities are tight, which yields the following:  
\begin{enumerate}
\item[(a)]
$\delta^-_{S} (X'_i) = \delta^-_{S'} (X'_i) + 1 = k+1$, which implies that ${\rm head}(f'_1) = {\rm head}(f_1) \in X'_i$, ${\rm tail}(f'_1) \not\in X'_i$, and ${\rm tail}(f_1) \in X'_i$; see Figure~\ref{fig:figure3} (right), 
\item[(b)]
$X_i \cup X'_i$ is a tight set with respect to $S$, and 
\item[(c)]
$S$ contains no arc connecting $X_i \setminus X'_i$ and $X'_i \setminus X_i$. 
\end{enumerate}
By (a) and (b), $X_i \cup X'_i$ is a tight set with respect to $S$ that contains $f_1$, 
and hence $X_1 \subseteq X_i \cup X'_i$, because $X_1$ is the unique minimal tight set containing $f_1$. 
Therefore, any arc $e \in S$ contained in $X_1$ is also contained in $X_i \cup X'_i$. 
This together with (c) shows that such $e$ is contained in either $X_i$ or $X'_i$, which completes the proof. 
\end{proof}
\begin{figure}
    \centering
    \includegraphics[width=70mm]{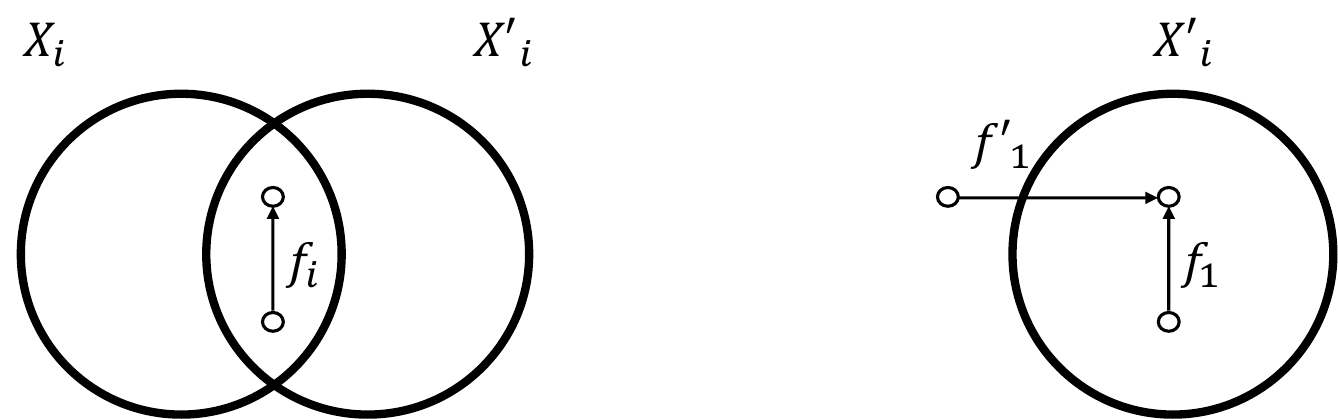}
    \caption{Location of the arcs $f_i$, $f_1$, and $f'_1$.}
    \label{fig:figure3}
\end{figure}

\subsection{Auxiliary Digraph}
\label{sec:auxiliary}

For two feasible arc subsets $S$ and $T$, we construct an associated \emph{auxiliary digraph} $H = (V_H, A_H)$ such that 
$V_H = [p]$ and $A_H$ contains an arc $(i, j)$ if $X_i$ contains $e_j$. 
Recall that $X_i \subseteq V$ is the inclusionwise minimal tight vertex set with respect to $S$ that contains $f_i$, 
or $X_i = V$ if no such tight set exists. 
Note that $H$ may contain self-loops. 
For example, in the case of Example~\ref{ex:detour} (see also Example~\ref{ex:Xi}), 
$H$ forms a dicycle of length $2$. 
In this subsection, we show some properties of $H$. 

\begin{lemma}\label{lem:04}
Every vertex in $H$ has at least one outgoing arc (possibly, a self-loop). 
\end{lemma}

\begin{proof}
Assume to the contrary that $i \in V_H = [p]$ has no outgoing arc in $H$. 
Then, by the definition of $H$, $X_i$ does not contain $e_j$ for any $j \in [p]$. 
Define $I^+, I^- \subseteq [p]$ as 
\begin{align*}
I^+ &= \{ j \in [p] \mid {\rm head}(e_j) = {\rm head}(f_j) \in X_i,\  {\rm tail}(e_j) \in X_i,\ {\rm tail}(f_j) \not\in X_i \}, \\ 
I^- &= \{ j \in [p] \mid {\rm head}(e_j) = {\rm head}(f_j) \in X_i,\  {\rm tail}(e_j) \not\in X_i,\ {\rm tail}(f_j) \in X_i \};  
\end{align*}
see Figure~\ref{fig:figure4}. 
Since $X_i$ contains $f_i$ but does not contain $e_i$, it holds that $i \in I^-$. 
We also see that $I^+ = \emptyset$ as $X_i$ does not contain $e_j$ for each $j$. 
Then, we obtain 
\[
k \le \delta^-_{T} (X_i) = \delta^-_{S} (X_i)  + |I^+| - |I^-| = k + |I^+| - |I^-| \le k-1,  
\]
which is a contradiction. 
\end{proof}

\begin{figure}[tbp]
    \centering
    \includegraphics[width=70mm]{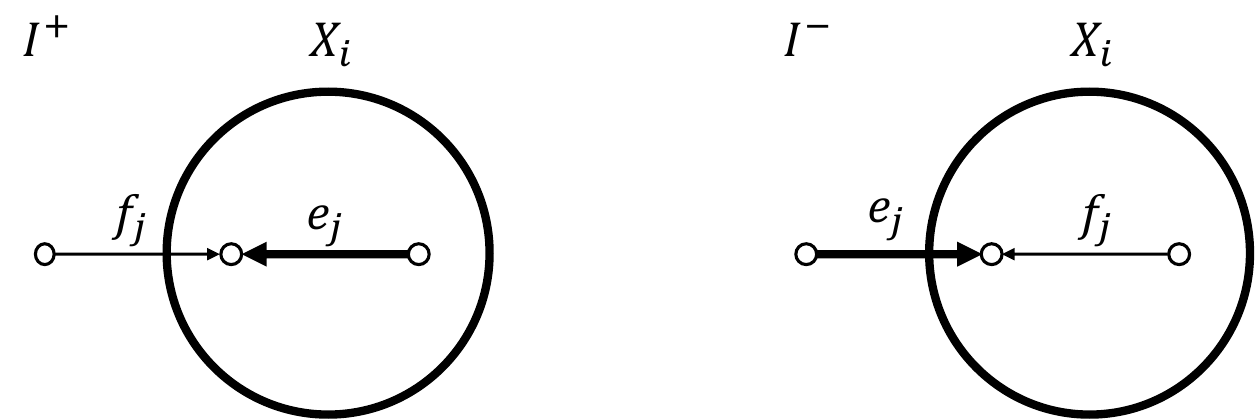}
    \caption{Definitions of $I^+$ and $I^-$.}
    \label{fig:figure4}
\end{figure}

\begin{lemma}\label{lem:05}
If $H$ has an arc $(i, j)$, then $X_i \cap X_j \neq \emptyset$. 
\end{lemma}

\begin{proof}
Since $X_i$ contains $e_j$ and $X_j$ contains $f_j$, both $X_i$ and $X_j$ contain the vertex ${\rm head}(e_j) = {\rm head}(f_j)$. 
This shows that $X_i \cap X_j \neq \emptyset$.
\end{proof}

\begin{lemma}\label{lem:06}
If $H$ has a dipath $P$ such that 
$\bigcup_{i \in V(P)} X_i$ contains some arc $e \in S$, then 
there exists $i \in V(P)$ such that $X_i$ contains $e$,  
where $V(P)$ denotes the set of vertices in $P$. 
\end{lemma}

\begin{proof}
If $X_i = V$ for some $i \in V(P)$, then the claim is obvious. 
Thus, it suffices to consider the case when $X_{i} \neq V$ for any $i \in V(P)$.
By renaming the indices if necessary, we may assume that $P$ traverses $1, 2, \dots, |V(P)|$ in this order. 
Assume to the contrary that $e$ is not contained in $X_i$ for any $i \in V(P)$. 
Let $1 \le i < j \le |V(P)|$ be indices that minimize $j-i$ subject to $X_i \cup X_{i+1} \cup \dots \cup  X_j$ contains $e$. 
Let $Y := X_i \cup X_{i+1} \cup \dots \cup X_{j-1}$. 
By applying Lemmas~\ref{lem:tight} and~\ref{lem:05} repeatedly, we see that $Y$ is tight and $Y \cap X_j \neq \emptyset$. 
Then, Lemma~\ref{lem:tight} shows that no arc in $S$ connects $Y \setminus X_j$ and $X_j \setminus Y$.  
Hence, $e \in S$ has to be contained in $Y$ or $X_j$, which contradicts the minimality of $j-i$. 
\end{proof}

\subsection{Shortest Dicycle}
\label{sec:shortestdicycle}

We see that $H$ has a dicycle  by Lemma~\ref{lem:04}. 
Let $C$ be a shortest dicycle in $H$, and let $q$ denote its length. 
If $q  = 1$, i.e., $H$ contains a self-loop incident to $i \in V_H$, 
then Lemma~\ref{lem:01} shows that $S' := S - e_i + f_i$ is feasible and $|S' \setminus T| = p-1$. 
Thus, in what follows, we consider the case when $q \ge 2$. 
This implies that  $X_{i} \neq V$ for any $i \in [p]$. 
By renaming the indices if necessary, we may assume that $C$ traverses $1, 2, \dots , q \in V_H$ in this order. 
Let $Y := X_2 \cup X_3 \cup \dots \cup X_q$. 
Note that both $Y$ and $X_1 \cap Y$ are tight with respect to $S$ 
by Lemmas~\ref{lem:tight} and~\ref{lem:05}.

\begin{lemma}\label{lem:071}
Arc $e_2$ is not contained in $Y$.
\end{lemma}

\begin{proof}
Assume to the contrary that $e_2$ is contained in $Y$. 
Then, by Lemma~\ref{lem:06}, there exists $i \in \{2, 3, \dots , q\}$ 
such that $X_i$ contains $e_2$. 
In such a case, since $H$ contains an arc $(i, 2)$ by definition, 
$H$ has a dicycle traversing $2, 3, \dots , i$ in this order, 
which contradicts the choice of $C$. 
\end{proof}

\begin{lemma}\label{lem:072}
Arc $e_2$ is from $X_1 \setminus Y$ to $X_1 \cap Y$. 
\end{lemma}

\begin{proof}
Since ${\rm head} (e_2) = {\rm head} (f_2) \in X_2 \subseteq Y$, 
Lemma~\ref{lem:071} shows that ${\rm tail} (e_2) \not\in Y$. 
We also see that $e_2$ is contained in $X_1$ as $H$ contains arc $(1, 2)$. 
By combining them,  $e_2$ is from $X_1 \setminus Y$ to $X_1 \cap Y$.
\end{proof}

\begin{lemma}\label{lem:08}
Arc $f_1$ is from $X_1 \setminus Y$ to $X_1 \cap Y$. 
\end{lemma}

\begin{proof}
By definition, $f_1$ is contained in $X_1$, which means that ${\rm head} (f_1) \in X_1$ and ${\rm tail} (f_1) \in X_1$. 
Furthermore, since $e_1$ is contained in $X_q$ as $H$ contains arc $(q, 1)$, 
we have that ${\rm head} (f_1) = {\rm head} (e_1) \in X_q \subseteq Y$.
Thus, it suffices to show that ${\rm tail} (f_1) \not\in Y$. 

Assume to the contrary that ${\rm tail} (f_1) \in Y$.
Then, $X_1 \cap Y$ contains $f_1$. 
Since $X_1 \cap Y$ is a tight set with respect to $S$ by Lemma~\ref{lem:tight} and 
$X_1 \cap Y \subseteq X_1 \setminus \{ {\rm tail} (e_2) \}$ by Lemma~\ref{lem:072}, this contradicts the minimality of $X_1$.  
\end{proof}

See Figure~\ref{fig:figure1} for the illustration of Lemmas~\ref{lem:071}--\ref{lem:08}.

\begin{figure}[tbp]
 \begin{tabular}{cc}
 \begin{minipage}[t]{0.5\hsize}
  \begin{center}
   \includegraphics[width=45mm]{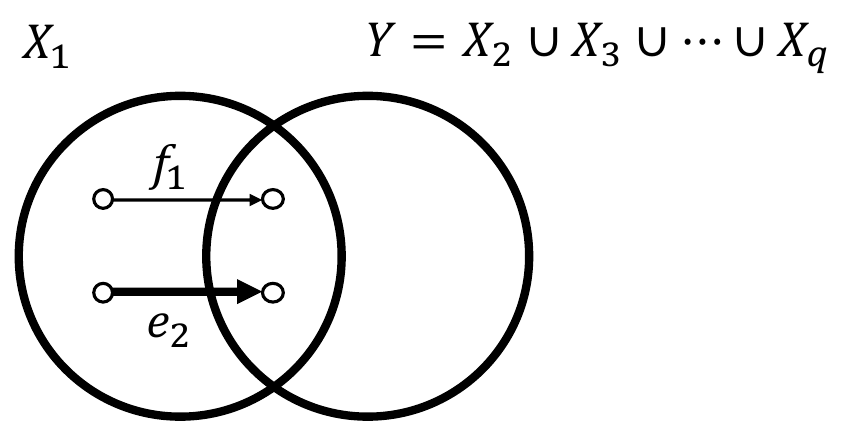}
   \caption{Location of $e_2$ and $f_1$.
}
  \label{fig:figure1}
   \end{center}
 \end{minipage}
 \hfill
 \begin{minipage}[t]{0.5\hsize}
   \begin{center}
   \includegraphics[width=45mm]{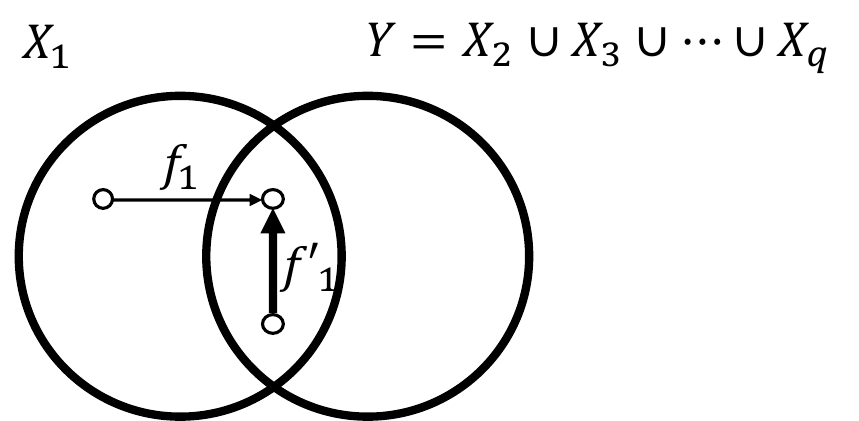}
   \caption{Location of $f'_1$.}
     \label{fig:figure5}
    \end{center}
 \end{minipage}\\
 \end{tabular}
\end{figure}

\begin{lemma}\label{lem:09}
There exists an arc $f'_1 \in S$ with ${\rm head}(f'_1) = {\rm head}(f_1)$ such that either 
\begin{enumerate}
\item
$f'_1 \in S \setminus T$ and $f'_1$ is contained in $X_1$, or 
\item
$f'_1 \in S \cap T$ and $f'_1$ is contained in $X_1 \cap Y$. 
\end{enumerate}
\end{lemma}

\begin{proof}
If ${\rm head}(e_2) = {\rm head}(f_1)$, then $f'_1 := e_2$ satisfies the first condition. 
Thus, suppose that ${\rm head}(e_2) \neq {\rm head}(f_1)$. 
Then, since $\delta^-_S(X_1 \cap Y) = k$, $\delta^-_S({\rm head}(f_1)) = k$, and $e_2 \in \Delta^-_S(X_1 \cap Y) \setminus \Delta^-_S({\rm head}(f_1))$ by Lemma~\ref{lem:072}, 
we obtain $\Delta^-_S({\rm head}(f_1)) \setminus \Delta^-_S(X_1 \cap Y) \neq \emptyset$. 
Therefore, $S$ has an arc $f'_1$ with ${\rm head}(f'_1) = {\rm head}(f_1)$ such that
$f'_1 \not\in \Delta^-_S(X_1 \cap Y)$, which implies that $f'_1$ is contained in $X_1 \cap Y$ (Figure~\ref{fig:figure5}). 
Such an arc $f'_1$ satisfies one of the conditions. 
\end{proof}

Let $f'_1$ be an arc as in Lemma~\ref{lem:09} and let $S' := S - f'_1 + f_1$, which is feasible by Lemma~\ref{lem:01}. 
If $f'_1$ satisfies the first condition in the lemma (i.e., $f'_1 \in S \setminus T$), then $|S' \setminus T| = p-1$, and hence we are done. 
Thus, in what follows, we consider the case when $f'_1$ satisfies the second condition in the lemma. 
In this case, define $X'_i \subseteq V$ for each $i \in [p]$ 
as in Section~\ref{sec:keylemma}. 
Define the auxiliary digraph $H'$ associated with $S'$ and $T$ in the same way as $H$.

\begin{lemma}\label{lem:10}
Suppose that $f'_1$ satisfies the second condition in Lemma~\ref{lem:09}. 
Then, the auxiliary digraph $H'$ associated with $S'=S - f'_1 + f_1$ and $T$ has a dicycle of length at most $q -1$. 
\end{lemma}

\begin{proof}
We first show that $X_i \neq X'_i$ for some $i \in [q]$. 
Assume to the contrary that $X_i = X'_i$ for each $i \in [q]$. 
Then, we see that $Y = X_2 \cup \dots \cup X_q$   
is a tight set with respect to $S$, and we also see that $Y = X'_2 \cup \dots \cup X'_q$ is tight with respect to $S'$. 
This shows that $\delta^-_S(Y) = k = \delta^-_{S'}(Y)$.  
However, we obtain $\Delta^-_{S'}(Y) = \Delta^-_S(Y) \cup \{ f_1 \}$ by 
Lemma~\ref{lem:08} and by the second condition in Lemma~\ref{lem:09}, 
which is a contradiction.

Therefore, $X_i \neq X'_i$ for some $i \in [q]$. 
Let $i$ be the minimal index with $X_i \neq X'_i$, where we note that $i \ge 2$ by Lemma~\ref{lem:02}. 
Since $C$ is a shortest dicycle, $H$ does not contain an arc $(i, 2)$, that is, $X_i$ does not contain $e_2$. 
As $X_i \neq X_i'$ and $X_i$ does not contain $e_2$, by applying Lemma~\ref{lem:03} with $e=e_2$, 
we see that $X'_i$ contains $e_2$, which means that $H'$ contains an arc $(i, 2)$. 
By the minimality of $i$, $X'_j = X_j$ holds for $j \in \{2, \dots , i-1\}$, and hence 
$H'$ contains a dicycle $C'$ traversing $2, 3, \dots, i$ in this order. Since the length of $C'$ is at most $q-1$, this completes the proof. 
\end{proof}

\subsection{Putting Them Together}
\label{sec:together}

By the above lemmas, we obtain Proposition~\ref{prop:reducedifference} as follows.  
Suppose that $S \subseteq A$ and $T \subseteq A$ are feasible arc subsets and 
$H$ is the auxiliary digraph associated with $S$ and $T$.
If $H$ has a self-loop incident to $i \in V_H$, then $S' := S - e_i + f_i$ satisfies the conditions in Proposition~\ref{prop:reducedifference}. 
Otherwise, let $q$ be the length of a shortest dicycle in $H$ and let $f'_1 \in S$ be an arc satisfying the condition in Lemma~\ref{lem:09}. 
By the description just after Lemma~\ref{lem:09} and by Lemma~\ref{lem:10}, 
$S' := S - f'_1 + f_1$ satisfies the conditions in Proposition~\ref{prop:reducedifference}
or the auxiliary digraph $H'$ associated with $S'$ and $T$ has a dicycle of length at most $q -1$. 
Since the shortest dicycle length decreases monotonically, 
by applying such a transformation of $S$ at most $q$ times, we obtain a feasible arc subset $S'$ satisfying the conditions in Proposition~\ref{prop:reducedifference}.
This completes the proof of Proposition~\ref{prop:reducedifference}. 
Furthermore, by applying Proposition~\ref{prop:reducedifference} $p$ times, we obtain Theorem~\ref{thm:main}. 
Note that since all the proofs are constructive and each $X_i$ can be computed by using a minimum $s$-$t$ cut algorithm, the reconfiguration sequence can be computed in polynomial time.

\section{Upper Bound on the Sequence Length}
\label{sec:upperbound}

In this section, we give an upper bound on the length of a shortest reconfiguration sequence satisfying the conditions in Theorem~\ref{thm:main}. 
To this end, we first give an upper bound on the length of a shortest dicycle in the auxiliary graph $H$. 

\begin{lemma}\label{lem:dicyclelength} 
    The length $q$ of a shortest dicycle in the auxiliary graph $H$ (see Section~\ref{sec:shortestdicycle}) is at most 
    $\min\{|S \setminus T|,  k\}$. 
\end{lemma}

\begin{proof}
 We use the notation introduced in Section~\ref{sec:shortestdicycle}. 
 Since it is obvious that each dicycle in $H$ has length at most $|S \setminus T|$, it suffices to show that $q \le k$. 
 When $q=1$, the claim is obvious. Suppose that $q \ge 2$. 
 Let $Z = X_1 \cup \dots \cup X_q$. 
 Since $Y$ and $Z = X_1 \cup Y$ are tight sets with respect to $S$ and $e_2 \in \Delta^-_S(Y) \setminus \Delta^-_S(Z)$ by Lemma~\ref{lem:072}, 
 there exists an edge $e'_1$ in $\Delta^-_S(Z)\setminus \Delta^-_S(Y)$. 
 This means that $e'_1$ is from $V \setminus Z$ to $X_1 \setminus Y$. (Figure~\ref{fig:figure6}). 
 By the same argument, for $i\in [q]$, 
 $S$ has an arc $e'_i$ that is from $V \setminus Z$ to $X_i \setminus \bigcup_{j\neq i} X_j$. 
 Since $e'_1, \dots , e'_q$ are distinct arcs in $S$, we obtain $\delta^-_S(Z) \ge q$. 
 This shows that $q \le k$, as $Z$ is a tight set. 
\end{proof}

\begin{figure}
    \centering
    \includegraphics[width=60mm]{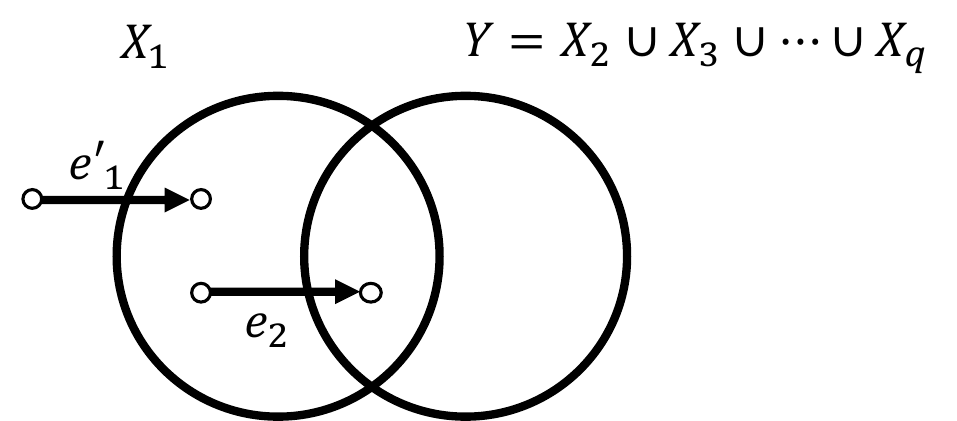}
    \caption{Arc $e'_1$ from $V \setminus Z$ to $X_1 \setminus Y$.}
    \label{fig:figure6}
\end{figure}

Using this lemma, we give an upper bound on the reconfiguration sequence. 

\begin{theorem}
\label{thm:length}
    There is a reconfiguration sequence satisfying the conditions in Theorem~\ref{thm:main} whose length is at most
    \[
    \begin{cases}
      \frac{p (p+1)}{2}  & \mbox{if $p \le k$,} \\
      \frac{k (k+1)}{2} + (p-k) k  & \mbox{if $p > k$,} \\      
    \end{cases}
    \]    
    where  $p=|S \setminus T|$. 
\end{theorem}

\begin{proof}
    As discussed in Section~\ref{sec:together}, we can decrease the value $|S \setminus T|$ by one by executing at most $q$ steps, 
    where $q$ is the length of a shortest dicycle in $H$ 
    and $q \le \min\{|S \setminus T|, k\}$ by Lemma~\ref{lem:dicyclelength}.   
    To obtain a sequence in Theorem~\ref{thm:main}, we apply this procedure $p$ times in which $S$ is replaced with an updated feasible solution $S'$.  
    Since $|S' \setminus T|$ takes the values $p, p-1, \dots , 1$, the total number of steps is at most $\sum_{i=1}^{p} \min\{i, k\}$, which is the desired value. 
\end{proof}

\section{Extension to Arborescences with Distinct Roots}
\label{sec:distinctroot} 

In this section, we prove Theorem~\ref{thm:main2}.
That is, we extend Theorem~\ref{thm:main} to the case when a feasible arc set 
is the union of $k$ arc-disjoint arborescences that may have distinct roots.

\begin{proof}[Proof of Theorem~\ref{thm:main2}]

Extend $V$ by adding a new vertex $\widehat{r}$. For an arc set $A' \subseteq A$ satisfying $\delta^{-}_{A'}(v) \le k$ for each $v \in V$, let $\widehat{A'}$ denote the arc set of the digraph obtained from $A'$ by adding $k-\delta^{-}_{A'}(v)$ parallel arcs from $\widehat{r}$ to $v$ for each $v \in V$.  Observe that $A' \in \mathcal{F}_k$ holds if and only if $\widehat{r}$ has outdegree $k$ in $\widehat{A'}$ and $\widehat{A'}$ can be partitioned into $k$ arc-disjoint $\widehat{r}$-arborescences on $V+\widehat{r}$. By Theorem~\ref{thm:edmonds}, the latter is equivalent to that $\delta^{-}_{\widehat{A'}}(X) \ge k$ for any $X \subseteq V$ with $X \ne \emptyset$.

We consider the case first when the multisets of roots in the decompositions of $S$ and $T$ into $k$ arc-disjoint arborescences are not the same.

\begin{lemma} \label{lem:changeroots}
    Suppose that $\delta^{-}_S(v) \ne \delta^{-}_T(v)$ holds for a vertex $v \in V$. Then there exist arcs $e \in S \setminus T$ and $f \in T \setminus S$ such that $S-e+f \in \mathcal{F}_k$.
\end{lemma}
\begin{proof}
Since $\sum_{w\in V} \delta^{-}_{S}(w) = \sum_{w \in V} \delta^{-}_{T}(w)$, there is a vertex $v \in V$ with $\delta^{-}_S(v) < \delta^{-}_T(v)$.  Then there exists an arc $f \in T\setminus S$ with ${\rm head}(f) = v$. Let $X$ denote the unique minimal subset of $V$ containing $f$ which is tight with respect to $\widehat{S}$, i.e., $\delta^{-}_{\widehat{S}}(X) = k$. Note that such a tight set exists as $\delta^{-}_{\widehat{S}}(V) = k$. Since $\delta^{-}_{\widehat{S}}(w) = \delta^{-}_{\widehat{T}}(w)=k$ for any $w \in V$ and $\delta^{-}_{\widehat{T}}(X) \ge k = \delta^{-}_{\widehat{S}}(X)$, 
\[\Delta_{T}(X,X) = k|X|-\delta^{-}_{\widehat{T}}(X) \le k|X| - \delta^{-}_{\widehat{S}}(X) = \Delta_{S}(X,X).\]
Therefore, $X$ contains an arc $e \in S \setminus T$, since it contains the arc $f \in T \setminus S$. 

We claim that $S' = S-e+f \in \mathcal{F}_k$. Let $u = {\rm head}(e)$, then
\[\widehat{S'} = \widehat{S}-e+(\widehat{r},u)+f-(\widehat{r},v).\]
We easily see that $\delta^{-}_{S'}(w) \le k$ holds for any $w \in V$, as $\delta^{-}_{S'}(v) \le \delta^{-}_S(v)+1 \le \delta^{-}_T(v) \le k$. Thus, it suffices to show that $\delta^{-}_{\widehat{S'}}(Z) \ge k$ holds for any nonempty subset $Z \subseteq V$.
Assume to the contrary that there exists a nonempty subset $Z \subseteq V$ with $\delta^{-}_{\widehat{S'}}(Z) \le k-1$. Then
\[k \le \delta^-_{\widehat{S}}(Z) \le \delta^-_{\widehat{S}-e+(\widehat{r},u)}(Z) \le 
\delta^-_{\widehat{S'}}(Z)+1 \le (k-1)+1 = k,\]
thus $ \delta^{-}_{\widehat{S}}(Z) = k$ and  $\delta^{-}_{\widehat{S}}(Z) = \delta^{-}_{\widehat{S}-e+(\widehat{r},u)}(Z) = \delta^{-}_{\widehat{S'}}(Z)+1$. These show that $Z$ is a tight set with respect to $\widehat{S}$, it does not contain $e$, and it contains $f$. Since $X \cap Z$ is a tight set with respect to $\widehat{S}$ by Lemma~\ref{lem:tight} and $X \cap Z \subsetneq X$ as $Z$ does not contain $e$, this contradicts the minimality of $X$.  
\end{proof}

We turn to the proof of the theorem. By the repeated application of Lemma~\ref{lem:changeroots}, there exists a sequence $T_0, T_1, \dots, T_m$ such that $T_0=S$, $\delta^-_{T_m}(v) = \delta^-_T(v)$ for any $v \in V$, $T_i \in \mathcal{F}_k$ for $i\in [m]\cup\{0\}$ and $|T_{i-1}\setminus T_i| = |T_i \setminus T_{i-1}| = 1$ for $i \in [m]$. By Theorem~\ref{thm:main}, there exists a sequence $T'_m$, $T'_{m+1}, \dots, T'_\ell$ such that $T'_m = \widehat{T_m}$, $T'_\ell = \widehat{T}$, $T'_i$ is a subset of $\widehat{T'_m}\cup\widehat{T}$ which can be partitioned into $k$ arc-disjoint $\widehat{r}$-arborescences on $V+\widehat{r}$ for $i \in \{m, m+1, \dots, \ell\}$, and $|T'_{i-1}\setminus T'_i| = |T'_i \setminus T'_{i-1}| = 1$ for $i \in \{m+1,m+2,\dots, \ell\}$. Then for any $i \in \{m+1,m+2,\dots, \ell\}$ there is an arc set $T_i \subseteq A$ such that $T'_i = \widehat{T_i}$, as $\delta^-_{T'_i}(v) = k$ for any $v \in V$. Since $T'_i \subseteq \widehat{T'_m} \cup \widehat{T}$ and $T'_i$ can be partitioned into $k$ arc-disjoint $\widehat{r}$-arborescences, $\widehat{r}$ has outdegree $k$ in $T'_i$, and $T_i$ can be partitioned into $k$ arc-disjoint arborescences. Therefore, $T_i \in \mathcal{F}_k$ holds for $i \in \{m+1, m+2, \dots, \ell\}$, hence the sequence $T_0,T_1,\dots,T_\ell$ satisfies the properties required by the theorem.
\end{proof}

The above proof shows that the length of a shortest reconfiguration sequence in Theorem~\ref{thm:main2} 
has the same upper bound as in Theorem~\ref{thm:length}.

\section{Proof of Theorem~\ref{thm:sumcounterexample}}
\label{sec:sumcounterexample}

In this section, we give a proof of Theorem~\ref{thm:sumcounterexample}, which we restate here. 

\sumcounterexample*

\begin{proof}
Consider the matroids $M_1 = (E,\mathcal{B}_1)$ and $M_2 = (E, \mathcal{B}_2)$ on ground set $E = \{a,b,c_1,c_2,c_3,d_1,d_2,d_3\}$ defined by their families of bases
\begin{align*}
\mathcal{B}_1 & = \{B \subseteq E \mid |B| = 3, |B\cap \{c_1,c_2,c_3\}| \le 1, |B \cap \{d_1,d_2,d_3\}| \le 1\}, \\
\mathcal{B}_2 & = \{B \subseteq E \mid |B| = 3, |B \cap \{a, c_1, d_1\}| = 1\}.
\end{align*}
Note that $M_1$ is the truncation of the direct sum of the uniform matroids of rank 1 on $\{a\}$, $\{b\}$, $\{c_1,c_2,c_3\}$, and $\{d_1,d_2,d_3\}$, while $M_2$ is the direct sum of the uniform matroid of rank 1 on $\{a,c_1,d_1\}$ and the uniform matroid of rank 2 on $\{b,c_2,c_3,d_2,d_3\}$.

We prove that the pair $(M_1,M_2)$ satisfies \ref{it:rcb}. The common bases of $M_1$ and $M_2$ are the sets of the form
\[\{a,b,c_i\}, \{a,b,d_j\}, \{a,c_i,d_j\}, \{b,c_1,d_j\}, \{b,c_i,d_1\}\]
for $i,j\in \{2,3\}$. It is enough to show the existence of a reconfiguration sequence between $\{b,c_1,d_2\}$ and each $B\in \mathcal{B}_1 \cap \mathcal{B}_2$. For $i,j\in\{2,3\}$, consider the sequence of common bases
\[\{b,c_1,d_2\}, \{b,c_1,d_j\}, \{a,b,d_j\}, \{a,c_i,d_j\}, \{a,b,c_i\}, \{b,c_i,d_j\},\]
where we omit the second term for $j=2$. 
This sequence starts from $\{b,c_1,d_2\}$, contains each $B \in \mathcal{B}_1 \cap \mathcal{B}_2$ for appropriate values of $i, j \in \{2,3\}$, and $|B' \setminus B''| = |B'' \setminus B'| = 1$ holds for each pair of adjacent terms $B', B''$ of the sequence, thus it proves our claim.

Next we show that the matroids $2M_1 = (E,\mathcal{B}^2_1)$ and $2M_2  = (E,\mathcal{B}^2_2)$ do not satisfy \ref{it:rcb}. 
Recall that $2 M_i$ is the matroid whose bases are the unions of two disjoint bases of $M_i$. 
We have
\begin{align*}
    \mathcal{B}_1^2 & = \{B \subseteq E \mid \{a, b\} \subseteq B, |B \cap \{c_1, c_2, c_3\}| = 2, |B \cap \{d_1, d_2, d_3\}| = 2\}, \\
    \mathcal{B}_2^2 & = \{B \subseteq E \mid |B \cap \{a, c_1, d_1\}| = 2, |B \cap \{b,c_2,c_3,d_2,d_3\}| = 4\},
\end{align*}
thus \[\mathcal{B}_1^2 \cap \mathcal{B}_2^2 = \{\{a,b,c_1,c_i,d_2,d_3\} \mid i \in \{2,3\}\} \cup \{\{a,b,c_2,c_3,d_1,d_j\} \mid j \in \{2,3\}\}.\]
Since $|B \setminus B'| = |B' \setminus B| = 2$ for any $B \in \{\{a,b,c_1,c_i,d_2,d_3\} \mid i \in \{2,3\}\}$ and $B' \in \{\{a,b,c_2,c_3,d_1,d_j\} \mid j \in \{2,3\}\}$, the pair $(2M_1, 2M_2)$ does not satisfy \ref{it:rcb}. 
\end{proof}

\section{Exact Arborescence Packing Problem}
\label{sec:exactpoly}

In this section, as an application of Theorem~\ref{thm:main2}, we present a polynomial-time algorithm 
for the \emph{exact arborescence packing problem}. 
Suppose we are given a digraph $D=(V, A)$ whose every arc is colored either red or blue, 
and we are also given positive integers $k$ and $p$. 
Let $\mathcal{F}_{k} \subseteq 2^A$ denote the family of all arc subsets that 
can be partitioned into $k$ arc-disjoint arborescences.
The \emph{exact arborescence packing problem} asks for 
an arc set $T \in \mathcal{F}_{k}$ that has exactly $p$ red arcs if one exists.
We obtain the following result as a corollary of Theorem~\ref{thm:main2}. 
Note that a similar result is obtained for $\mathcal{F}_{k,r}$ by using Theorem~\ref{thm:main}. 
When $k=1$, this corresponds to a result of Barahona and Pulleyblank~\cite{barahona1987exact}.

\begin{theorem}
    \label{cor:exactpoly}
    The exact arborescence packing problem can be solved in polynomial time. 
\end{theorem}

\begin{proof}
Define a weight function $w\colon A \to \{0, 1\}$ as follows: $w(e)=0$ if $e$ is blue, and $w(e)=1$ if $e$ is red. 
For an arc subset $T \subseteq A$, we denote $w(T) = \sum_{e \in T} w(e)$. 
Since $\mathcal{F}_{k}$ is represented as the intersection of basis families of two matroids, 
we can find an arc set $T_{\rm min}$ (resp.~$T_{\rm max})$ in $\mathcal{F}_k$ that minimizes (resp.~maximizes) the total weight
by using a weighted matroid intersection algorithm (see e.g.,~\cite{lexbook}). 

If $w(T_{\rm min}) > p$ or $w(T_{\rm max}) < p$, then there exists no arc set $T \in \mathcal{F}_k$ with $w(T)=p$. 
Therefore, we can conclude that a desired arc set does not exist. 

We now show that a desired arc set can be found if $w(T_{\rm min}) \le p \le w(T_{\rm max})$. 
By Theorem~\ref{thm:main2}, we can find a sequence 
$T_0, T_1, \dots , T_\ell$ such that $T_0 = T_{\rm min}$, $T_\ell = T_{\rm max}$,  
$T_i \in \mathcal{F}_{k}$ for $i \in [\ell] \cup \{0\}$, and $|T_{i-1} \setminus T_i| = |T_{i} \setminus T_{i-1}| = 1$ for $i \in [\ell]$. 
Since $|w(T_{i-1}) - w(T_i)| \le 1$ for each $i \in [\ell]$, if $w(T_{\rm min}) \le p \le w(T_{\rm max})$, then 
there exists an index $i^* \in [\ell] \cup \{0\}$ such that $w(T_{i^*}) = p$. 
This means that $T_{i^*}$ has exactly $p$ red arcs. 

Therefore, we can find a desired arc set $T \in \mathcal{F}_k$ in polynomial time if one exists. 
\end{proof}

\section{Hardness of Matroid Intersection Reconfiguration}
\label{sec:MIRHardness}

In this section, we show that 
both {\sc RCB Testing} and {\sc Matroid Intersection Reconfiguration} require an exponential number of queries if the matroids are given by independence oracles.
We show the hardness of these two problems simultaneously. 
Note that the following theorem obviously implies Theorem~\ref{thm:hardRCB}.  

\begin{theorem}
    \label{thm:hardness}
    {\sc RCB Testing} and {\sc Matroid Intersection Reconfiguration} require an exponential number of independence queries
    if the matroids are given by independence oracles.
\end{theorem}

\begin{proof}
    Let $r\ge 2$ be an even integer and  $A = \{a_1, \dots, a_r\}$ and $B=\{b_1, \dots, b_r\}$ be two disjoint sets of size $r$, and let $E=A \cup B$. 
    Let $M_1=(E,\mathcal{B}_1)$ denote the partition matroid defined by 
    \[\mathcal{B}_1 = \{B \subseteq E \mid |B \cap \{a_1, b_1\}| = \dots = |B \cap \{a_r, b_r\}| = 1\}.\]
    Consider the sparse paving matroid $M_2=(E,\mathcal{B}_\mathcal{H})$ defined by Lemma~\ref{lem:paving} with \[\mathcal{H} = \{H \subseteq E  \mid |H \cap \{a_1, b_1\}| = \dots = |H \cap \{a_r, b_r\}| = 1, |H \cap A| = r/2\}.\] 
    That is, $\mathcal{B}_\mathcal{H} = \{B' \subseteq E \mid |B'|=r, B' \not \in \mathcal{H}\}$. 
    It is not difficult to check that $\mathcal{H}$ satisfies the condition of Lemma~\ref{lem:paving}, that is, $|H \cap H'| \le r-2$ holds for each $H, H' \in \mathcal{H}$, $H \ne H'$.

    \begin{claim} \label{cl:paving1}
    There is no reconfiguration sequence of common bases of $M_1$ and $M_2$ from $A$ to $B$. In particular, the pair $(M_1, M_2)$ does not satisfy \ref{it:rcb}. 
    \end{claim}
    \begin{proof}
    Assume to the contrary that there exists a sequence $B_0, \dots, B_\ell$ of common bases such that $A=B_0$, $B=B_\ell$ and $|B_i \setminus B_{i-1}| = |B_{i-1} \setminus B_i| = 1$ for $i \in [\ell]$. Since $|A \cap B_0| = 0$,  $|A \cap B_\ell| = r$ and $||B_{i-1}\cap A|-|B_i\cap A|| \le 1$ for each $i \in [\ell]$, there is an index $j \in [\ell]$ such that $|A \cap B_j| = r/2$. Then $B_j \in \mathcal{H}$, which is a contradiction, since $B_j$ is a basis of $M_2$.
    \end{proof}

    For a set $H_0\in \mathcal{H}$, consider the matroid $M'_2=(E, \mathcal{B}_{\mathcal{H}\setminus \{H_0\}})$ defined by Lemma~\ref{lem:paving}.
    That is, $\mathcal{B}_{\mathcal{H}\setminus \{H_0\}} = \mathcal{B}_{\mathcal{H}} \cup \{H_0\}$. 

    \begin{claim} \label{cl:paving2}
        There is a reconfiguration sequence of common bases of $M_1$ and $M'_2$ from $A$ to $B$.
        Furthermore, the pair $(M_1, M'_2)$ satisfies \ref{it:rcb}.
    \end{claim}
    \begin{proof}
        Let $\{i_1, \dots, i_r\} = \{1,\dots, r\}$ be such that $H_0 = \{b_{i_1}, \dots, b_{i_{r/2}}, a_{i_{r/2+1}}, \dots, a_{i_r}\}$. Then $B_j = \{b_{i_1}, \dots, b_{i_j}, a_{i_{j+1}}, \dots, a_{i_r}\}$ is a common basis of $M_1$ and $M'_2$ for $j \in [r] \cup \{0\}$, $B_0=A$, $B_r = B$, and $|B_{j-1}\setminus B_j| = |B_j\setminus B_{j-1}| = 1$ for $j \in [r]$. Therefore, $A$ is reconfigurable to $B$. 
        Observe that any common basis $B'$ of $M_1$ and $M'_2$ with $|B' \cap A| \le r/2$ is reconfigurable to $A$ in a greedy way. 
        Similarly, any common basis $B'$ of $M_1$ and $M'_2$ with $|B' \cap A| > r/2$ is reconfigurable to $B$. 
        By these observations and the fact that $A$ is reconfigurable to $B$, we see that the pair $(M_1, M'_2)$ satisfies \ref{it:rcb}.
    \end{proof}

    We are ready to prove the theorem. Consider any algorithm solving {\sc RCB Testing} or {\sc Matroid Intersection Reconfiguration}. By Claims~\ref{cl:paving1} and \ref{cl:paving2}, the algorithm answers differently for the pairs of matroids $(M_1, M_2)$ and $(M_1, M'_2)$, 
    where the input common bases are $A$ and $B$ if we consider {\sc Matroid Intersection Reconfiguration}. 
    Since the only independence query distinguishing $M_2$ and $M'_2$ is the query of the independence of $H_0$, the algorithm for $M_1$ and $M_2$ must query the independence of each $H_0 \in \mathcal{H}$. Therefore, any such algorithm uses at least $|\mathcal{H}| = \binom{r}{r/2}$ independence queries.
\end{proof}

\section{Concluding Remarks}
\label{sec:conclusion}

In this paper, we showed the reconfigurability of the union of $k$ arborescences for fixed $k$. 
In other words, we showed that the pair of matroids representing the union of $k$ arborescences satisfies~\ref{it:rcb}.  
It will be interesting to investigate whether~\ref{it:rcb} holds or not for other classes of matroid pairs, e.g., White's conjecture~\cite{white1980unique}. 

Another interesting topic is the length of a shortest reconfiguration sequence. 
For the union of $k$ arborescences, 
in Section~\ref{sec:upperbound}, we give an upper bound on the length of a shortest reconfiguration sequence, which is slightly smaller than $k |S \setminus T|$. 
Meanwhile, there is an example whose shortest length is $\frac{3}{2} |S \setminus T|$, which is obtained by combining many copies of the digraph in Example~\ref{ex:detour}. 
It will be interesting if we can close the gap between these bounds. 
It is also open whether we can find a shortest reconfiguration sequence from $S$ to $T$ in polynomial time if $S$ and $T$ are given as input. 

The length of a shortest reconfiguration sequence can be considered also for other classes of matroid pairs. 
When $M_2$ is the dual matroid of $M_1$, Hamidoune conjectured that there always exists a reconfiguration sequence whose length is at most the size of each common basis (or equivalently, the rank of the matroids); see~\cite{cordovil1993bases}. 
This conjecture is stronger than White's conjecture~\cite{white1980unique}, and is open even for some special cases, e.g.~when $M_1$ is a graphic matroid and $M_2$ is its dual.

\section*{Acknowledgement}
This work was supported by the Research Institute for Mathematical Sciences, an International Joint Usage/Research Center located in Kyoto University  and by the Lend\"{u}let Programme of the Hungarian Academy of Sciences -- grant number LP2021-1/2021.
The authors thank members of the project ``Fusion of Computer Science, Engineering and Mathematics Approaches for Expanding Combinatorial Reconfiguration'' for discussion on this topic. The authors are grateful to Andr\'as Frank for bringing the paper \cite{barahona1987exact} to their attention.
Yusuke Kobayashi was supported by JSPS KAKENHI Grant Numbers 
JP20H05795, 
JP20K11692, and  
JP22H05001. 
Tam\'{a}s Schwarcz was supported by the \'{U}NKP-22-3 and \'{U}NKP-23-3  New National Excellence Programs of the Ministry for Culture and Innovation from the source of the National Research, Development and Innovation Fund.

\bibliographystyle{plain}
\bibliography{main}

\end{document}